\documentclass{article}
\usepackage{graphicx}
\usepackage{xcolor,fullpage,amsmath,amsthm,amsfonts,amssymb,bbold,braket,hyperref,thm-restate}
\usepackage[capitalize, nameinlink]{cleveref}

\usepackage[ruled,vlined]{algorithm2e}
\DontPrintSemicolon
\SetKwInOut{Input}{Input}
\SetKwInOut{Output}{Output}
\SetKwProg{Fn}{Function}{ is}{end}

\newtheorem{lemma}{Lemma}

\newtheorem{theorem}{Theorem}
\newtheorem*{theorem*}{Theorem}
\newtheorem{corollary}{Corollary}

\newtheorem{definition}{Definition}
\newtheorem*{problem*}{}

\newcommand{\NP}{\mathsf{NP}}
\newcommand{\DTIME}{\mathsf{DTIME}}

\newcommand{\FF}{\mathbf{F}}

\newcommand{\NCP}{\textsc{NCP}}
\newcommand{\CosetHeavy}{\textsc{Coset\-Heavy}}
\newcommand{\MultGapNCP}{\textsc{Mult\-Gap\-NCP}}
\newcommand{\OptHalfSparsifier}{\textsc{Opt\-Half\-Sparsifier}}
\newcommand{\ApproxOHS}{\textsc{Approx\-OHS}}

\newcommand{\calC}{\mathcal{C}}

\newcommand{\wt}{\text{wt}}
\newcommand{\minwt}{\text{minwt}}

\title{On the Hardness of the One-Sided Code Sparsifier Problem}

\date{\today}

\author{
Elena Grigorescu\thanks{David R. Cheriton School of Computer Science, 
University of Waterloo, Canada. 
 \url{elena-g@uwaterloo.ca}} 
\and
Alice Moayyedi \thanks{David R. Cheriton School of Computer Science, 
University of Waterloo, Canada. 
\url{anelosima@proton.me}}}

\begin{document}

\maketitle

\begin{abstract}
    The notion of \emph{code sparsification} was introduced by Khanna, Putterman and Sudan (SODA 2024), as an analogue to the the more established notion of cut sparsification in graphs and hypergraphs. In particular, for $\alpha\in (0,1)$ an (unweighted) one-sided $\alpha$-sparsifier for a linear code $\calC \subseteq \FF_2^n$ is a subset $S\subseteq [n]$ such that the weight of each codeword projected onto the coordinates in $S$ is preserved up to an $\alpha$ fraction.  Recently, Gharan and Sahami (arxiv.2502.02799) show the existence of one-sided $\frac12$-sparsifiers of size $n/2+O(\sqrt{kn})$ for any linear code,  where $k$ is the dimension of $\calC$. In this paper, we consider the computational problem of finding a one-sided $\frac12$-sparsifier of minimal size, and show that it is NP-hard, via a reduction from the classical nearest codeword problem. We also show hardness of approximation results.
\end{abstract}

\section{Introduction}\label{sec:intro}

For $\alpha\in (0,1)$, a one-sided $\alpha$-\emph{sparsifier} for a code $\calC \subset \FF_2^n$ is a set $S\subseteq [n]$ such that the projection of any codeword $c \in \calC$ onto $S$ results in a vector $c_S$ whose weight is preserved up to an $\alpha$-factor, namely that $wt(c_S)\geq \alpha \cdot wt(c)$. Here $wt(c)=|\{i \mid c_i\ne 0\}|$.

The notion of (weighted) two-sided code sparsifiers was recently introduced by Khanna, Putterman, and Sudan~\cite{KhannaPS24}, as an analogue of the notion of cut sparsifiers for graphs and hypergraphs~\cite{Karger94}, which have been studied even more broadly in the context of constraint satisfaction problems~\cite{KoganK15}. Recently, Gharan and Sahami~\cite{GharanSahami2025} study the unweighted one-sided code sparsifier problem for linear codes (i.e.\ subspaces). They give an elegant short proof of the existence of one-sided $\frac12$-sparsifiers of size $n/2+O(\sqrt{kn})$, where $k=\log|C|$ is the dimension of $C$.

Here we study the computational problem of finding one-sided unweighted $\frac12$-sparsifiers, defined as follows.

\begin{quote}
    \underline{Minimal One-Sided $\frac12$-Sparsifier Problem ($\OptHalfSparsifier$)}\\
    \textbf{Instance:} A linear code $\calC$,  given by its generators.\\
    \textbf{Output:} A set $S \subseteq [n]$ such that:
    \begin{itemize}
    \item (Feasibility) for all $c \in \calC$, $\wt(c_S) \geq \frac{1}{2}c$;
    \item (Optimality) $S$ is of smallest size among all sets that satisfy the above.
    \end{itemize}
\end{quote}

To the best of our knowledge, the complexity of finding minimum-sized sparsifiers has not been studied before. Here we show the following hardness results.

\begin{theorem}[
        restate=thmmain,
        name=Hardness of $\OptHalfSparsifier$
        ]\label{thm:main}
    $\OptHalfSparsifier$ is $\NP$-hard.
\end{theorem}

In fact, our results hold more generally:

\begin{theorem}[
        restate=thmapprox,
        name=Approximation Hardness of $\OptHalfSparsifier$]\label{thm:approx}
    Let $S^*\subseteq[n]$ be a one-sided $\frac12$-sparsifier of minimal size for a linear code $\calC$. The problem of finding a one-sided $\frac12$-sparsifier $S \subseteq [n]$ for $\calC$ such that $\gamma \cdot |\bar{S}| \geq |\bar{S^*}|$ (where $\bar{S}=[n]/S$) is:
    \begin{itemize}
    \item $\NP$-hard for any constant $\gamma \geq 1$;
    \item impossible to solve in polynomial time for any constant $\epsilon > 0$, assuming $\NP \nsubseteq \DTIME(2^{\log^{O(1)}n})$ and with $\gamma = 2^{\log^{1-\epsilon}n}$;
    \item impossible to solve in polynomial time for some constant $c > 0$, assuming $\NP \nsubseteq \bigcap_{\delta > 0} \DTIME(2^{n^\delta})$ and with $\gamma = n^{c/\log\log n}$.
    \end{itemize}
\end{theorem}

Our proofs show Turing reductions from the fundamental problem of computing a nearest codeword to a received string, as defined below.

\begin{quote}
    \underline{Nearest Codeword Problem (\NCP)}\\
    \textbf{Instance:} A linear code $\calC$ given by its generators, a received string $s \in \FF_2^n$, and an integer $k$.\\
    \textbf{Output:} (YES) if there exists $c \in \calC$ such that $\wt(c + s) \leq k$, and (NO) otherwise.
\end{quote}

The NP-hardness of the nearest codeword problem was first shown by Vardy~\cite{vardy1997intractability}, followed by proofs by Dumer, Miccancio, and Sudan~\cite{dumer2003approx} of the hardness of the problem for promise additive and multiplicative approximation versions, under RUR reductions. A sequence of follow-ups~\cite{Khot05,AustrinKhot,Micciancio2014, 
BhattiproluLee24, bhattiproluGLR2025} have now established that the multiplicative variant is NP-hard under deterministic Karp reductions.

\subsection{Preliminaries}\label{sec:prelims}

Let $\calC$ be a \emph{linear} code over $\FF_2^n$; that is, $\calC$ is a subset of $\FF_2^n$ such that $c, c' \in \calC \Rightarrow c + c' \in \calC$. We may define a linear code $\calC$ as the span of the columns of a matrix $M$; in this case, we call $M$ the \emph{generator matrix} of $\calC$, and we call the columns of $M$ the \emph{generators} of $\calC$. We say that $\calC$ has \emph{dimension} $k$ if $|\calC| = 2^k$. For $h\in\FF_2^n$, $h+\calC$ is called and \emph{affine} subspace, or a \emph{coset} of $\calC$ in $\FF_2^n$. 

We denote by ${\bf 0}$ and ${\bf 1}$ the all-zeroes and the all-ones vectors, respectively, in $\FF^n$.

We define the \emph{weight} of a string $s \in \FF_2^n$, $\wt(s)$, as the number of nonzero coordinates of $s$. For a set of coordinates $S \subseteq [n]$ and a string $c \in \FF_2^n$, we $c_S$ as the projection of $c$ onto the coordinates in $S$.

\begin{definition}
A set $S \subseteq [n]$ is \textbf{$\alpha$-thin} with respect to $\calC$ if for every codeword $c \in \calC$,

\[\wt(c_S) \leq \alpha \cdot \wt(c).\]
\end{definition}

Likewise, we call a set $S$ \textbf{$\alpha$-thick} with respect to $\calC$ if $\wt(c_S) \geq \alpha \cdot \wt(c)$ for all $c \in \calC$. If we define identify the set $S$ with its indicator vector $s \in \FF_2^n$, the weight $\wt(c \circ s)$ is equal to the weight $\wt(c_S)$, where $\circ$ is the Schur or element-wise product. We say that a string $s$ is $\alpha$-thin or $\alpha$-thick with respect to a code $\calC$ exactly when its corresponding set $S$ is. Note that the complement of an $\alpha$-thin set or string is an $\alpha$-thick set or string, and vice versa. We denote by ${\bar S}$ the complement of the set $S$ in $[n]$.

The terminology \textit{$\alpha$-thin} is by analogy to $\alpha$-thinness in the context of graphs. An $\alpha$-thin subgraph (usually, a tree) of a graph $G$ is a subgraph $T \subseteq G$ such that, for any cut $\delta(S)$ of $G$, the number of edges in $T$ which are in $\delta(S)$ is at most an $\alpha$ fraction of the total number of edges in $\delta(S)$. This notion corresponds directly to the above definition of $\alpha$-thinness in linear codes via the following relation:

Given a graph $G = (V,E)$, we define the incidence matrix $M$ as a $|V|\times |E|$ matrix such that for $v \in V, e \in E$,

\[M_{i,e} =
    \begin{cases}
        1&\text{if } v\in e\\
        0&\text{if } v\notin e
    \end{cases}
\]

That is, $M_{v,e} = 1$ exactly when $e$ is incident on $v$, and $M_{v,e}=0$ otherwise. If we use $M$ as the generator matrix for a linear code $C(M)$, then the $\alpha$-thin subgraphs of $G$ correspond to the $\alpha$-thin sets of $E$ with respect to $C(M)$. If we define $\alpha$-thick subgraphs as the complements of $\alpha$-thin subgraphs, these also correspond to the $\alpha$-thick sets of $E$.

There is also a direct correspondence between this notion of $\alpha$-thickness and the notion of a one-sided $\alpha$-sparsifier. If a set is $\alpha$-thick with respect to a code $\calC$, it is also a one-sided $\alpha$-sparsifier for that code, and vice versa.

In \cref{sec:main}, we use the following hardness of approximation theorem for the NCP problem~\cite{vardy1997intractability,dumer2003approx,Khot05,
AustrinKhot,Micciancio2014,BhattiproluLee24, bhattiproluGLR2025}.

\begin{theorem}\label{thm:ncphard}
Given an affine subspace $V \subseteq \FF_2^n$ and an integer $k>0$, there is no polynomial-time algorithm which distinguishes between the following cases:

\begin{itemize}
    \item (YES) there exists $x \in V$ with $\wt(x) \leq k$.
    \item (NO) for all $x \in V$, it is the case that $wt(x) \geq \gamma \cdot k$.
\end{itemize}

\begin{enumerate}
    \item when $\gamma>1$ is a constant, assuming $P\neq \NP$
    \item when $\gamma=2^{\log^{1-\epsilon}n}$,  for any $\epsilon$, assuming $\NP\not \subseteq \DTIME(2^{\log^{O(1)}n})$
    \item when  $\gamma = n^{c/\log\log n}$, for some $c>0$, assuming 
    $\NP \nsubseteq \bigcap_{\delta > 0} \DTIME(2^{n^\delta})$ and with $\gamma = n^{c/\log\log n}$.
\end{enumerate}
\end{theorem}

\subsection{Organization}\label{sec:organization}
In \cref{sec:structure} we give a structural theorem of $\frac12$-thick sets. We use this to show that the problem of finding $\frac12$-thick sets is strongly related to the nearest codeword problem. In \cref{sec:main} we reduce $\NCP$ to the problem of finding optimal $\frac12$-thick sets (equivalently, optimal one-sided $\frac12$-sparsifiers) through this relationship, proving the main theorems.

\section{Representatives of Largest Weight}\label{sec:structure}
In this section we characterise $\frac12$-thick sets, and connect them to the nearest codewords of elements in the same coset.

For a linear code $\calC \subseteq \FF_2^n$, we consider the quotient space $\FF_2^n/\calC$. The equivalence classes or cosets $H \in \FF_2^n/\calC$ are the sets of strings such that $h, h' \in H \Rightarrow h + h' \in \calC$ and $h \in H, s \notin H \Rightarrow h + s \notin \calC$. For a given coset $H$, we define the set $H^* \subseteq H$ to be the set of strings in $H$ of largest weight; $H^* = \{h^* \in H : \wt(h^*) = \max\limits_{h\in H}\wt(h)\}$. Our first theorem is a characterisation of the $\frac12$-thick strings in a given coset:

\begin{theorem}\label{thm:structure}
    For a member $h$ of a coset $H \in \FF_2^n/\calC$, the following are equivalent:
    \begin{enumerate}
    \item $h$ is among the elements of $H$ of greatest weight;
    \item $h$ is $\frac12$-thick with respect to $\calC$;
    \item The all-zeroes string, $\mathbf{0}$, is a nearest codeword to the complement of $h$, $\bar{h}$ (that is, $h + \mathbf{1}$).
    \end{enumerate}
\end{theorem}

We prove \cref{thm:structure} in two parts: \cref{lem:12eq} (equivalence of 1.\ and 2.) and \cref{lem:13eq} (equivalence of 1.\ and 3.).

We say that $c \in \calC$ is a nearest codeword to a string $s \in \FF_2^n$ when the Hamming distance between $s$ and $c$, $\wt(s + c)$, is minimal among all elements of $\calC$. In general, we denote the complement of a string $s$ as $\bar{s}$, and the all-zeroes and all-ones strings as $\mathbf{0}$ and $\mathbf{1}$ respectively. Note that $\bar{s} = s + \mathbf{1}$.

\begin{lemma}\label{lem:12eq}
    1.\ and 2.\ above are equivalent.
\end{lemma}
\begin{proof}
  \cite{GharanSahami2025} presents a proof that 1.\ implies 2., which we will briefly reproduce here for completeness. Suppose that $h^* \in H^*$ is a string of maximal weight among strings in $H$, and suppose that $h^*$ is not $\frac12$-thick with respect to $\calC$; that is, $\exists c \in \calC$ such that $\wt(c \circ h^*) < \frac{1}{2}\wt(c)$. Take the string $h' = h^* + c$, noting that $h' \in H$. We must have that $\wt(h') > \wt(h^*)$, since $c\circ h^* < c\circ \bar{h^*}$ --- adding $c$ turns more coordinates of $h^*$ to one than it turns to zero. This contradicts the assumption that $h^*$ is of maximal weight among strings in $H$; hence, $h^*$ is $\frac12$-thick with respect to $\calC$.

    Now, to show that 2.\ implies 1., suppose that $h \in H$ is a $\frac12$-thick string, and that there is some $h^* \in H$ such that $\wt(h^*) > \wt(h)$. Then consider the codeword $c = h^* + h$, with $c \in \calC$. Since $c + h = h^*$, and $\wt(h^*) > \wt(h)$, we must have that $c \circ h < c \circ \bar{h}$; hence, $\wt(c \circ h) < \frac{1}{2}\wt(c)$. Then $h$ is not $\frac12$-thick, a contradiction; any $\frac12$-thick string must be of maximal weight among its equivalence class.
\end{proof}

\begin{lemma}\label{lem:13eq}
    1.\ and 3.\ above are equivalent.
\end{lemma}
\begin{proof}
    To show that 1.\ implies 3., take some $h^* \in H$ of maximal weight, and suppose that there exists some codeword $c \in \calC$ such that $\wt(c + \bar{h^*}) < \wt(\bar{h^*} + \mathbf{0}) = \wt(\bar{h^*})$. Then, noting that for any $a \in \FF_2^n$, $\wt(a) = n - \wt(\bar{a}) = n - \wt(a + \mathbf{1})$, we have that: 
    \[
        \wt(c + \bar{h^*}) < \wt(\bar{h^*})
        \quad\text{iff}\quad
        n - \wt(\mathbf{1} + c + \bar{h^*}) < n - \wt(h^*)
        \quad\text{iff}\quad
        \wt(c + h^*)> \wt(h^*).
    \]
    
    Hence $c + h^*$ is an element of $H$ of greater weight than $h^*$, a contradiction. More intuitively, if $\bar{h^*}$ is closer to $c$ than to $\mathbf{0}$, it must be the case that it shares more coordinates with $c$ than it has zero coordinates. Then $h^*$ must differ on more coordinates with $c$ than it has nonzero coordinates --- so $h^* + c$, the string consisting of all coordinates in which $h^*$ and $c$ differ, must be of higher weight than $h^*$ itself. So $\mathbf{0}$ must be a nearest codeword in $\calC$ to $\bar{h^*}$.

    To show that 3.\ implies 1., suppose that $h \in H$ is of less than maximal weight; that is, there is some $h^* \in H$ with $\wt(h^*) > \wt(h)$. Then $c = h^* + h$ is a codeword in $\calC$ which is closer to $\bar{h}$ than $\mathbf{0}$:

    \[
        \wt(\bar{h^*}) < \wt(\bar{h})
        \quad\text{iff}\quad
        \wt((h^* + h) + \bar{h}) < \wt(\bar{h})
        \quad\text{iff}\quad
        \wt(c + \bar{h})< \wt(\bar{h})
    \]

    So $\mathbf{0}$ is not a nearest codeword to $\bar{h}$; if $\mathbf{0}$ is a nearest codeword to some $\bar{h^*}$ with $h^* \in H$, $h^*$ is of maximal weight among elements of $H$.
\end{proof}

This concludes the proof of \cref{thm:structure}.

The fundamental property at hand here is that the operation of shifting by some constant string is an isometry; once the nearest (or furthest) codewords of one element $h$ of a coset $H$ have been determined, the nearest codewords of every other element $h'$ in $H$ are fully determined by the difference $h$ and $h'$. Thus, since any element in the coset $H$ must have some nearest codewords, there must be some elements of $H$ which have any particular codeword as their nearest codeword --- if $h \in H$ has $c$ as a nearest codeword, $h + (c + c')$ has $c'$ as a nearest codeword. The elements of $\mathbf{1} + H$ which are closest to $\mathbf{0}$ --- corresponding to the elements of $H$ which are furthest from $\mathbf{0}$ --- must therefore not merely be closer to $\mathbf{0}$ than any other element of $\mathbf{1} + H$, but also closer to $\mathbf{0}$ than they are to any other element of $\calC$. This allows us to extract from $H^*$ knowledge of the nearest codewords of every element of $\mathbf{1} + H$:

\begin{theorem}\label{thm:ncpeq}
    Let $H \in \FF_2^n/\calC$ be some equivalence class and $H^*$ be the set of strings of maximal weight in $H$. Consider any string $h \in H$ and any string $h^* \in H^*$. Their sum, $h + h^*$, is a codeword in $\calC$ of minimum Hamming distance from $\bar{h}$. Likewise, for any $h$, the codewords $C$ of minimum Hamming distance from $\bar{h}$ are of the form $h+h^*$ for some $h^* \in H^*$.
\end{theorem}
\begin{proof}
    Take any $h \in H$ and $h^* \in H^*$, and suppose that there exists some string $c \in \calC$ such that $\wt(\bar{h}+c) < \wt(\bar{h}+(h+h^*))$. Then:
    
    \[
        n - \wt(h + c) < n - \wt(h^*)
        \quad\text{iff}\quad
        \wt(h + c)>\wt(h^*).
    \]
    
    Since $h + c$ is in $H$, and elements of $H^*$ have maximum weight among elements of $H$, this is a contradiction; no such $c$ can exist. Furthermore, $\wt(\bar{h} + (h + h^*)) = n - \wt(h^*)$, which is the same quantity for all $h^* \in H$; hence, all codewords of the form $h+h^*$ have equal and minimal Hamming distance from $\bar{h}$. This proves the first direction. For the second direction, suppose that there exists a codeword $c in \calC$ of minimum Hamming distance from $\bar{h}$; that is, $\wt(\bar{h} + c) \leq \wt(\bar{h} + (h+h^*))$ for any $h^* \ in H^*$. By the first direction, this is an equality; $\wt(\bar{h} + c) = \wt(\bar{h} + (h + h^*))$, since if $\wt(\bar{h} + c) < \wt(\bar{h} + (h + h^*))$ then $h + h^*$ would not be a nearest codeword to $\bar{h}$. Then $\wt(h + c) = \wt(h^*)$; as $h + c \in H$, it must be the case that $h + c \in H^*$. So $c = h + (h + c) = h + h^*$ for some $h^* \in H^*$. This proves the second direction.

    Alternatively, we may observe that since $\mathbf{0}$ is a nearest codeword to $\bar{h^*}$, $\mathbf{0}+c = c$ must be a nearest codeword to $\bar{h^*} + c$; that is, if $c = h^* + h$, then $c$ is a nearest codeword to $\bar{h^*} + (h^* + h) = \bar{h}$. The inverse direction follows straightforwardly by reversing the argument.
\end{proof}

This also demonstrates that, for any elements $h, h'$ in some coset $H$ of $\FF_2^n/\calC$, the nearest codewords in $\calC$ to $h$ are the same distance from $h$ as the nearest codewords to $h'$. We can therefore talk about the ``distance'' of a coset from the code; the distance of a coset $H$ to $\calC$ is the distance from any element of $H$ to its nearest codewords in $\calC$. Since $\mathbf{0}$ is always a codeword in any linear code, and the elements of minimum weight in a given coset must have $\mathbf{0}$ as a nearest codeword, the distance of any coset is equal to the minimum weight among elements in that coset, $\minwt(H)$.

Since the $\frac12$-thick strings in a coset directly correspond to the nearest codewords of elements in that coset, we have the following:

\begin{corollary}\label{cosetheavynp}
    Given a linear code $\calC\subseteq \FF^n$,  and integer $k$, the problem of determining whether a string in a given coset of $\calC$ in $\FF_2^n$ exists of weight at least $k$, or equivalently whether the weight of the $\frac12$-thick elements in a given coset are at least $k$, is $\NP$-complete, under deterministic Karp reductions.
\end{corollary}
\begin{proof}
This problem is obviously in $\NP$; we can certify any YES instance with a string of weight at least $k$ in the given coset. We show $\NP$-hardness by reduction from the nearest codeword problem, as defined in \cref{sec:intro}.

Let {\CosetHeavy} be the problem above:

\begin{quote}
    \underline{Heaviest Element Problem (\CosetHeavy)}\\
    \textbf{Instance:} A generator matrix $M$ for a linear code $\calC$, a string $h$ in an equivalence class $H \in \FF_2^n/\calC$, and an integer $k$.\\
    \textbf{Output}: (YES) if there exists $h^* \in H$ such that $\wt(h^*) \geq k$, and (NO) otherwise.
\end{quote}

We define the mapping from {\NCP} to {\CosetHeavy} instances as follows:

\[(M, s, k) \rightarrow (M, \bar{s}, n-k)\]

This map can obviously be computed in polynomial time.

Suppose that the {\NCP} instance $(M, s, k)$ is a (YES) instance; there exists some $c \in \calC$, with $\calC$ the linear code generated by $M$, such that $\wt(c + s) \leq k$. Then, specifically, the distance between $c + s$ and $\mathbf{0}$ --- $\wt(c + s + \mathbf{0})$ --- is also at most $k$. So, following \cref{thm:ncpeq}, there exists a string in the coset containing $s$ of $\calC$ in $\FF_2^n$ with weight at most $k$; there exists a string in the coset containing $\bar{s}$ with weight at least $n - k$. So the {\CosetHeavy} instance $(M, \bar{s}, n-k)$ is a (YES) instance.

Now suppose that the {\CosetHeavy} instance is a (YES) instance. Then there is a string in the coset $H$ containing $s$ with weight at most $k$, and for any element of $H$ there is a codeword in $\calC$ of distance at most $k$; the {\NCP} instance is a (YES) instance.

Thus, {\NCP} $\leq_p$ {\CosetHeavy}, and {\CosetHeavy} is $\NP$-complete.
\end{proof}

Note that this reduction is in fact surjective, and ${\NCP} \equiv_p {\CosetHeavy}$.

\section{Optimal One-Sided Sparsifiers (Proof of Main Theorems)}\label{sec:main}
Of course, in order to sparsify a code, we are not actually interested in finding the $\frac12$-thick strings among a particular coset. Instead,  we are interested in finding the $\frac12$-thick strings among all strings in $\FF_2^n$. Specifically, we are interested in the problem of finding the smallest strings which are $\frac12$-thick. We call a string an \textit{optimal} $\alpha$-thick (thin) string with respect to a code $\calC$ if it is $\alpha$-thick (thin) with respect to $\calC$ and it is of least (resp.\ greatest) weight among all $\alpha$-thick (thin) strings. Similarly, we call a set an optimal $\alpha$-thick set if it is of smallest size among such sets.

\begin{corollary}
    The optimal $\frac12$-thick strings with respect to a code $\calC$ are exactly the complements of the strings of greatest weight which have $\mathbf{0}$ as a nearest codeword in $\calC$.
\end{corollary}
\begin{proof}
By \cref{thm:structure}, the $\frac12$-thick strings with respect to a code $\calC$ are exactly the complements of the strings which have $\mathbf{0}$ as a nearest codeword --- the $\frac12$-thin strings with respect to $\calC$. Further, the $\frac12$-thick strings of least weight are the complements of the $\frac12$-thin strings of greatest weight.
\end{proof}

In fact, since the distance to the nearest codeword is constant among all elements of a given coset of $\calC$ in $\FF_2^n$, the cosets containing the optimal $\frac12$-thin strings are those where every element is of greatest distance to their nearest codewords. The cosets containing the optimal $\frac12$-thick strings, then, are cosets obtained by adding $\mathbf{1}$ to the cosets containing the optimal $\frac12$-thin strings (note that in a code containing the codeword $\mathbf{1}$, the cosets containing the optimal $\frac12$-thin strings are the same as the cosets containing the optimal $\frac12$-thick strings).

\thmmain*
\begin{proof}
We demonstrate this by polynomial-time Turing reduction from {\NCP}. We restate the problem for convenience:

\begin{quote}
    \underline{Minimal One-Sided $\frac12$-Sparsifier Problem ($\OptHalfSparsifier$)}\\
    \textbf{Instance:} A linear code $\calC$ given by its generators.\\
    \textbf{Output:} A set $S \subseteq [n]$ such that:
    \begin{itemize}
    \item (Feasibility) for all $c \in \calC$, $\wt(c_S) \geq \frac{1}{2}c$;
    \item (Optimality) $S$ is of smallest size among all sets which satisfy the above.
    \end{itemize}
\end{quote}

and use the following algorithm:

\begin{algorithm}[H]
    \SetKwFunction{OHS}{ohs}
    \Input{A generator matrix $M$ for a linear code $\calC$;\\
    a received string $s \in \FF_2^n$;\\
    an integer $k$;\\
    a subroutine \OHS{$M'$} which solves {\OptHalfSparsifier}, with output in the form of the indicator vector of the produced set.}
    \Output{Whether the {\NCP} instance $(M, s, k)$ is a (YES) instance or a (NO) instance.}
    $M_0 \leftarrow M$\;
    $i \leftarrow 0$\;
    \nl\While{a decision has not been made}{\label{alg1:loop}
        Let $\calC_i$ denote the linear code generated by the matrix $M_i$.\;
        $h^*_i \leftarrow$ \OHS{$M_i$}\;
        \nl\If{$n - \wt(h^*_i) \leq k$}{\label{alg1:if1}
        Output (YES).\;
        }
        \nl\If{$\bar{h^*_i} + s \in \calC_i$}{\label{alg1:if2}
        Output (NO).\;
        }
        $a_i \leftarrow \bar{h^*_i} + s$\;
        $M_{i+1} \leftarrow (a_i \text{ concatenated to } M_i)$\;
        $i \leftarrow i + 1$\;
    }
    \caption{Solving {\NCP} using {\OptHalfSparsifier}}\label{alg1}
\end{algorithm}

This algorithm terminates in at most $n-\text{dim}(\calC)$ calls to the {\OptHalfSparsifier} subroutine, with $O(n-k)$ extra work: since $a_i \notin \calC_i$ for all $i$, the dimension of $\calC_{i+1}$ is one greater than the dimension of $\calC_i$ for all $i$. When the dimension of $\calC_i$ is $n$, then it must be the case that $s \in \calC_i$; thus, the loop at \autoref{alg1:loop} is run at most $n-\text{dim}(\calC)$ times.

It remains to show correctness. We wish to maintain the invariant that, for each $i$, the distance from $s$ to the nearest codeword in $\calC_i$ is no closer than the distance from $s$ to the nearest codeword in $\calC$. So, towards a proof by contradiction, suppose that there is some $i$ such that $\calC_{i+1}$ has an element nearer to $s$ than the nearest element in $\calC_i$. Note that since $\calC_{i+1} = \calC_i \cup (a_i + \calC_i)$, if $\calC_{i+1}$ has an element closer to $s$ than all elements in $\calC_i$, that element must be in $a_i + \calC_i$. So we have:

\begin{align*}
    &\minwt(s + C(M_i)) > \minwt(s + a_i + \calC_i)\\
    \text{iff}\quad &\minwt(s + C(M_i)) > \minwt(s + \bar{h^*_i} + s + \calC_i)\\
    \text{iff}\quad &\minwt(s + C(M_i)) > \minwt(\bar{h^*_i} + \calC_i)
\end{align*}

Where, for $S \subseteq \FF_2^n$, $\minwt(S)$ is the smallest weight among elements of $S$. Thus, the smallest weight among elements of the coset $s + \calC_i$ is larger than the smallest weight among elements of the coset $\bar{h^*_i} + \calC_i$. But then the largest element in the coset $\bar{s}+\calC_i$ is smaller than the largest element in the coset $h^*_i$. By \cref{thm:structure}, all largest elements of any coset are $\frac12$-thick, and the largest element of $\bar{s}+\calC_i$ is $\frac12$-thick; thus, $h^*_i$ cannot be a $\frac12$-thick string of minimal weight among all $\frac12$-thick strings, a contradiction. So $\calC_{i+1}$ cannot have any elements nearer to $s$ than the nearest elements in $\calC_i$; by induction, $\calC_i$ has no closer elements to $s$ than does $\calC$ for any $i$.

We can see this property more intuitively by noting that adding a basis element to a linear code combines pairs of cosets which differ by that element; the maximal (minimal) elements in the resulting coset will be the maximal (minimal) elements among the two. Thus, given that the coset containing $\bar{h^*}$ contains a largest minimal element among any coset, the coset $(\bar{h^*} + \calC_i) \cup (\bar{s}+\calC_i)$ of $\calC_{i+1}$ will contain the minimal element of $\bar{s}+\calC_i$ as a minimal element; repeatedly merging cosets in this manner never moves the coset containing $s$ ``closer'' to the code.

Given this invariant, we proceed to show that the two conditionals on \autoref{alg1:if1} and \autoref{alg1:if2} can only be satisfied if the instance is a yes or no instance respectively. We begin with \autoref{alg1:if1}. Since $\bar{h^*_i}$ is a string of greatest distance from any codeword in $\calC_i$, and has $\mathbf{0}$ as a nearest codeword, if $\wt(\bar{h^*_i}) = n - \wt(h_i^*) \leq k$, then the greatest distance from any string among elements in $\calC_i$ is at most $k$; thus, there must be a codeword in $\calC_i$ at least that close to $s$. By the invariant above, there also must be a codeword in $\calC$ which is that close, and the {\NCP} instance $(M, s, k)$ is a yes instance. To show that the conditional on \autoref{alg1:if2} is satisfied only for no instances, note that if $\bar{h_i^*} + s \in \calC_i$, then $s$ is in a coset with $\bar{h_i^*}$. This implies that the distance from $s$ to its nearest codewords in $\calC_i$ is the same as that of $\bar{h_i^*}$ to its nearest codewords; since $\bar{h_i^*}$ has $\mathbf{0}$ as a nearest codeword, the distance from $s$ to its nearest codeword is exactly $\wt(\bar{h_i^*}) = n - \wt(h_i^*)$. Since the conditional on \autoref{alg1:if1} was not satisfied, then, there exists no codeword in $\calC_i$ which has distance at most $k$ from $s$. Since $\calC_i \supseteq \calC$, there exists no codeword in $\calC$ of at most that distance, and the {\NCP} instance $(M, s, k)$ is a no instance.

We therefore have a polynomial-time Turing reduction from  {\NCP} to {\OptHalfSparsifier}, and {\OptHalfSparsifier} is $\NP$-hard.
\end{proof}

From known hardness of approximation results (\cite{bhattiproluGLR2025})  of $\NCP$, we can also derive similar results for the $\OptHalfSparsifier$ problem.

\thmapprox*
\begin{proof}
This follows directly from the algorithm above and the hardness of approximation of {\NCP} given by \cref{thm:ncphard}. Suppose that instead of an algorithm which solves {\OptHalfSparsifier}, we have an algorithm which solves the following approximation of {\OptHalfSparsifier}:
\begin{quote}
    \underline{$\gamma$-Approximate Minimal One-Sided $\frac12$-Sparsifier Problem ($\ApproxOHS_\gamma$)}\\
    \textbf{Instance:} A linear code $\calC$ given by its generators.\\
    \textbf{Output:} A set $S \subseteq [n]$ such that:
    \begin{itemize}
    \item (Feasibility) for all $c \in \calC$, $\wt(c_S) \geq \frac{1}{2}c$;
    \item (Approximate Optimality) if $S^*$ is a set which satisfies the above feasibility condition, then $\gamma \cdot|\bar{S}| \geq |\bar{S^*}|$
    \end{itemize}
\end{quote}

If $\gamma=1$, then this problem is $\OptHalfSparsifier$. Note that the multiplicative factor here is a constraint not on the size of the set itself, but the size of its conjugate; it is trivial to find a $\frac12$-thick set with a size within a factor of $2$ of that of the smallest $\frac12$-thick set, by simply taking every coordinate which is represented among codewords in $\calC$. It is easy to see from this that approximation up to a constant factor of the set size is trivial.

Given such a subroutine we can use \autoref{alg1} to approximate {\NCP}, solving the following problem:

\begin{quote}
\underline{Nearest Codeword Problem with Multiplicative Gap $\gamma$ ($\MultGapNCP_\gamma$)}\\
\textbf{Instance:} A generator matrix $M$ for a linear code $\calC$, a received string $s \in \FF_2^n$, and an integer $k$.\\
\textbf{Output}: (YES) if there exists $c \in \calC$ such that $\wt(c + s) \leq k$. (NO) if for every $c \in \calC$, $\wt(c + s) > \gamma k$.
\end{quote}

To show that this reduction goes through, we follow the previous proof with a small modification: instead of maintaining that the augmented codes $\calC_i$ each have no codewords closer to $s$ than the closest codewords in $\calC$, we maintain that each $\calC_i$ has no codewords outside of $\calC$ of distance smaller than $k$ to $s$. We note first that if $h^*_i + \calC_i$ has a distance from $\calC_i$ at most $k$, then we will have answered (YES) during that iteration of the loop at \autoref{alg1:loop}; we only proceed to add $h^*_i + s$ to the generator matrix if $h^*_i + \calC_i$ is further from $\calC_i$ than $k$. Thus, the set $(s + a_i + \calC_i) = (\bar{h^*_i} + \calC_i)$ has a distance from $\calC$ no smaller than $k$. Formally, $\minwt(s + a_i + \calC_i) = \minwt(\bar{h^*_i} + \calC_i) > k$.

The reduction from {$\MultGapNCP_\gamma$} to {$\ApproxOHS_\gamma$} then follows from the same arguments as in the proof of \cref{thm:main}. Briefly, if the {$\MultGapNCP_\gamma$} instance is a (YES) instance, then it will certainly never be the case that the distance from $s$ to $\calC_i$ is greater than $k$ for any $\calC_i$, so \autoref{alg1} will output (YES); if the {$\MultGapNCP_\gamma$} instance is a (NO) instance, then, given the invariant above, \autoref{alg1} may never output (YES), and must output (NO).
\end{proof}

\end{document}